\documentclass[conference]{IEEEtran}
\IEEEoverridecommandlockouts

\usepackage{cite}
\usepackage{amsmath,amssymb,amsfonts,amsthm}
\usepackage{physics}
\usepackage{algorithmic}
\usepackage{graphicx}
\usepackage{textcomp}
\usepackage{xcolor}
\usepackage{dsfont}
\usepackage{comment}
\usepackage{bbold}
\usepackage[acronym]{glossaries}
\usepackage{tipa}
\usepackage{cuted} 
\usepackage{tikz}
\newcommand\copyrighttext{%
	\footnotesize \textcopyright 2026 IEEE. Personal use of this material is permitted.
	Permission from IEEE must be obtained for all other uses, in any current or future
	media, including reprinting/republishing this material for advertising or promotional
	purposes, creating new collective works, for resale or redistribution to servers or
	lists, or reuse of any copyrighted component of this work in other works.}
\newcommand\copyrightnotice{%
	\begin{tikzpicture}[remember picture,overlay]
		\node[anchor=south,yshift=10pt] at (current page.south) 
		{\fbox{\parbox{\dimexpr\textwidth-\fboxsep-\fboxrule\relax}{\copyrighttext}}};
	\end{tikzpicture}%
}

\def\BibTeX{{\rm B\kern-.05em{\sc i\kern-.025em b}\kern-.08em
    T\kern-.1667em\lower.7ex\hbox{E}\kern-.125emX}}

\theoremstyle{definition}
\newcounter{thm}
\newtheorem{theorem}[thm]{Theorem}

\newtheorem{remark}[thm]{Remark}

\newcommand{\qa}{\text{\textopeno}}

\newacronym{avc}{AVC}{Arbitrarily Varying Channel}
\newacronym{caavc}{CAAVC}{Correlation-Assisted Arbitrarily Varying Channel}
\newacronym{dmc}{DMC}{Discrete Memoryless Channel}
\newacronym{bpsk}{BPSK}{Binary Phase Shift Keying}
\newacronym{tmsv}{TMSV}{Two-Mode Squeezed Vacuum}
\newacronym{bsc}{BSC}{Binary Symmetric Channel}

\begin{document}

\title{Entanglement Enabled Data Transmission over an Arbitrarily Varying Channel\\
\thanks{This work was financed by the Bavarian Ministry for Economic Affairs (StMWi) via the 6GQT project, the Federal Ministry of Research, Technology and Space (BMFTR) of Germany via grants 16KISR026, 16KIS2604, 16KIS1598K, 16KISQ039, 16KISQ093 and 16KISQ077 as well as in the programme of ``Souver\"an. Digital. Vernetzt.''. Joint project 6G-life, project identification number: 16KISK002, by the Bavarian Ministry for Science and Arts (StMWK) via the NeQuS project and by the DFG via grant NO 1129/2-1.}
}

\author{
    \IEEEauthorblockN{Janis N\"otzel \emph{(Member, IEEE)}, Florian Seitz}
    \IEEEauthorblockA{
        Emmy-Noether Gruppe Theoretisches Quantensystemdesign\\
        Technische Universit\"at M\"unchen\\
        janis.noetzel@tum.de, flo.seitz@tum.de
    }
}

\maketitle
\copyrightnotice

\begin{abstract}
Shared randomness is the central ingredient for stabilizing symmetrizable communication systems against arbitrarily varying jammers. Given the presence of the jammer, however, the question arises how this precious resource could have been distributed. Several works discuss the use of external sources for this task. In this work, we show, based on the most standard optical communication model, how the sender and receiver can employ entangled two-mode squeezed states to counter the jamming attack of an energy-limited jammer during the distribution phase when both the sender and jammer are allowed to use binary phase shift keying and two-mode squeezed vacuum states.
\end{abstract}

\begin{IEEEkeywords}
Arbitrarily varying channel, bosonic channel, entanglement, jamming, resilience
\end{IEEEkeywords}

\section{Introduction}
\glspl{avc} were introduced in \cite{bbt} and developed into the standard information-theoretic model for communication systems that are subject to jamming. The question whether the capacity $C$ of an \gls{avc} is positive or not turned out to be nontrivial \cite{ericsonAVC,elimination,csiszarNarayanPositivity}: As soon as an amount of randomness shared between sender and receiver is available that scales polynomially with the blocklength, the capacity $\bar C$ of the system is described by an information-theoretic quantity similar to Shannon's capacity formula. To describe $C$, the notion of \emph{symmetrizability} of an \gls{avc} got established \cite{ericsonAVC,csiszarNarayanPositivity} (see also \cite{ahlswede-blinovsky}). A symmetrizable \gls{avc} has zero capacity, and the capacity formula for a non-symmetrizable one equals that of the same \gls{avc} with shared randomness between sender and receiver. 

The investigation of \glspl{avc} in signal processing was initiated in \cite{csiszarNarayan}, where the received signal was modeled as $y= s+x+z$, with $s$ being the jamming signal, $x$ the transmitted signal, $z$ the additive white Gaussian noise ($s,x,z\in\mathbb C$). The analog of the model \cite{csiszarNarayan} in quantum signal processing is given by coherent states $|s\rangle,|x\rangle$ interfering on a beamsplitter with transmittivity $\eta$, with the received signal $|\sqrt{\eta}s+\sqrt{1-\eta}x+z\rangle$ being disturbed by classical Gaussian noise $z$ \cite{noetzel-isit2024}.

\glspl{avc} display super-activation \cite{minglaiSuperActivation} and discontinuity of $C$ \cite{minglaiDiscontinuity,noetzelDiscontinuity}. The symmetrizability condition is not computable \cite{bocheUncomputable}, but approximate solutions can be efficiently computed \cite{detectingsymmetrizability}. The presence of just any correlated bipartite source that is \emph{independent} from the jammer is sufficient to guarantee operation at $\bar C$ \cite{ahlswede-cai-correlated,boche13}. Their capacity under different types of jamming has been studied for colored Gaussian noise \cite{peregAVC} and native quantum tasks \cite{quantumAVC,fullyQuantumAVC}, with similar observations. In some cases, even exact thresholds for the amount of shared randomness are known \cite{jaggiAVC}. 

These results establish that shared (common) randomness is a critical ingredient for communication under jamming, and already extremely small amounts of shared randomness suffice to achieve the random capacity $\bar C$. 


However, the following core question could not previously be asked in a meaningful way: \emph{Assuming a symmetrizable \gls{avc}, and no preshared correlated randomness, how could sender and receiver start to communicate, how could they e.g. establish a source of shared randomness?}

In this work, we address this highly contradictory situation for the first time. We provide evidence that a solution can be found by utilizing quantum resources, namely shared entangled states. Most importantly however, \emph{the sharing of these states may take place under the attack of the jammer}.

We let sender and jammer use two coherent states, $|\alpha\rangle$ and $|-\alpha\rangle$ (corresponding to messages $x=0,1$ and jamming signals $s=0,1$, with $\alpha \in \mathbb R^+$, as in \gls{bpsk}, or else a \gls{tmsv} state ($x=2$ and $s=2$) with the photon number of its one marginal being equal to $\alpha^2$. Both signals get mixed at a $50:50$ beam-splitter, which is a standard quantum-optical communication model \cite{guhaBeamsplitter-1,guhaBeamsplitter-2}. The receiver applies homodyne detection with sign binarization, so that the effective channel becomes 
\begin{align}\label{def:bpsk-avc}
    w(y|s,x) = \left\{\begin{array}{ll}
                    b_p(y|x),&x\ne2, s=x\\
                    b_{\tilde p}(y|x),&x\ne2, s=2\\
                    b_{\tilde p}(y|s),&x=2,s\ne2\\
                    \tfrac{1}{2},&\mathrm{else}
                \end{array}\right.
\end{align}
where $p=\tfrac{1}{2} (1-\mathrm{erf}(2\alpha))$ and $\tilde p=\tfrac{1}{2}(1-\mathrm{erf}(\alpha(1+\alpha^2)^{-1/2}))$ \cite[eq. (89) and (124)]{brask2022gaussian}, $p<\tilde p<1/2$ and $b_t$ is a \gls{bsc}.
Then $w$ is symmetrizable with the strategy $u(s|\hat x)=\delta(s,\hat x)$ for $\hat x=0,1,2$:
\begin{align}
    \textstyle\sum_su(s|\hat x)w_p(y|s,x)=\textstyle\sum_su(s| x)w_p(y|s,\hat x)\quad \forall x, \hat x,y.
\end{align}
Therefore, the channel cannot be used to communicate. The situation changes, however, if we allow the sender to perform a measurement on his half of the (entangled) \gls{tmsv} states transmitted by him and condition his later actions on the outcome.
Concretely, the sender may transmit the \gls{tmsv} state during $k/2$ out of $n$ transmissions to generate a distribution $q(u,v|\tilde s)$ shared between sender and receiver, where $\tilde s$ models the impact of the jammer. Then, \emph{all} effective channels
\begin{align}\label{def:augmented-avc}
    \tilde w_p(y|s,\tilde s,x)=\sum_{u,v}q(u,v|\tilde s)w_p(y\oplus v|s,x\oplus u)
\end{align}
are \glspl{bsc} with crossover probabilities smaller than $\frac{1}{2} -\delta$ for some $\delta>0$ that depends only on the allowed energy $\alpha^2$. 

Afterwards, they utilize $x\oplus u$ and $y\oplus v$ instead of $x$ and $y$ during the generation of common randomness for another $k/2$ rounds. Finally, $n-k$ rounds can be used to transmit at rates approaching the common randomness-assisted capacity of $w_p$. This leads to the surprising result that specific bipartite correlated signals, and in particular high-energy squeezed states, can increase the \emph{resilience} of certain conventional data transmission systems.

In what follows we first fix notation and provide the essential definitions that describe the behavior of \glspl{avc} and quantum optical systems. We then proceed to describe the quantum-optical correlation distribution procedure. Finally, we state and prove our main result.
\section{Preliminaries}
\subsection{Notation}
    We denote sets in boldface as $\mathbf X$. The set of probability distributions on a (finite) set is written as $\mathcal P(\mathbf X)$. For every $\mathbf X$, the uniform distribution on $\mathbf X$ is denoted as $\pi$. 

    Channels mapping a finite set $\mathbf X$ to another finite set $\mathbf Y$ are given by conditional probability distributions $w(y|x)$. Consequently, \glspl{avc} are given by conditional probability distributions $w(y|s,x)$ where $s$ is the input of the jammer. An important subclass of channels are the \glspl{bsc}, where $\mathbf X=\mathbf Y=\{0,1\}$ and the conditional probability distribution satisfies $w(y|x)=w(x|y)$. A \gls{bsc} is defined by its crossover probability $t$. We abbreviate its conditional probability distribution as $b_t$, with the convention $b_t(0|0)=1-t$. $\mathbb{P}(E)$ denotes the probability of the event $E$. The mutual information of two random variables $U,V$ is denoted $I(U,V)$. The convex hull of $p_1,\ldots,p_K\in\mathcal P(\mathbf X)$ is written $\mathrm{conv}(\{p_1,\ldots,p_K\})$.

\subsection{Arbitrarily Varying Channels}
An \gls{avc} is a family of channels $w(y|s,x)$ dependent on a channel state $s \in \mathbf S$ controlled by a jammer, where $\mathbf S$ is called the channel state set and $x, y$ are from the input and output alphabets $\mathbf X$ and $\mathbf Y$. 
A code of blocklength $n$ and rate $R$ consists of a distribution $q\in\mathcal P(\mathbf U\times\mathbf V)$ on finite sets $\mathbf U,\mathbf V$, for each $u\in\mathbf U$ a set of $M$ code-words $x^n_1(u),\ldots,x^n_M(u)\in\mathbf X^n$ and for each $v\in\mathbf V$ corresponding decoding sets $D_1(v),\ldots,D_M(v)\subset\mathbf Y^n$. For a state sequence $s^n\in\mathbf S^n$, the average probability of error is 
{\setlength{\abovedisplayskip}{4pt}
\setlength{\belowdisplayskip}{4pt}
\begin{align}
\scalebox{0.97}{$\displaystyle P_e(s^n) = 1-\tfrac{1}{M}\sum_{m=1}^{M}\sum_{u,v}q(u,v) w^{\otimes n}(D_m(v)|s^n,x^n_m(u)).$}
\end{align}}
We distinguish three types of distributions $q$ in this work: First, the case where $\mathbf U=\mathbf V=\{1\}$ and the code is called deterministic. Second, the case where $\mathbf U=\mathbf V$ and $q(u,v)=|\mathbf U|^{-1}\delta(u,v)$ and the code is called common randomness-assisted. Third, the case where $q$ is an i.i.d. distribution, and the code is called correlation-assisted. In all three cases, a rate $R$ is achievable if there exists a sequence of respective codes such that $\max_{s^n\in\mathbf S^n}P_e(s^n)\to 0$ as $n\to\infty$ and $\liminf_{n\to\infty}\tfrac{1}{n}\log M_n\geq R$. The respective capacity ($C$ for deterministic codes, $\bar C$ for common randomness-assisted codes and $\tilde C$ for correlation-assisted codes) is the supremum of all achievable rates.
An \gls{avc} is called symmetrizable if there exists a distribution $u(s|x)\in\mathcal P(\mathbf S)$ such that
\begin{align}
    \sum_s u(s|x)w(y|s,x') = \sum_s u(s|x')w(y|s,x)
\end{align}
for all $x,x'\in\mathbf X$ and $y\in\mathbf Y$. Intuitively, the jammer can simulate a second sender, making it impossible for the receiver to distinguish between codewords. If an \gls{avc} is symmetrizable, then $C=0$; otherwise $C=\bar C$ \cite{ahlswede70}. Further if $I(U,V)>0$ then $\bar C=\tilde C$ \cite{ahlswede-cai-correlated}.

\subsection{Quantum Optics}
In this section, we provide a short review of the most important concepts of quantum optics that are used in this work \cite{holevo-book}. For Gaussian states and operators, we mostly follow the conventions from \cite{brask2022gaussian}.

\subsubsection*{Continuous Variable Quantum Mechanics}
A photonic field mode is described as a harmonic oscillator with Hamiltonian
\begin{align}
\hat{H} = \hbar \omega \left( \hat a^{\dagger}\hat a + \frac{1}{2} \right).
\end{align}
Its eigenstates are the photon number states $\ket{n}$,
\begin{align}
\hat H\ket{n} = \hbar \omega \left( n+\frac{1}{2} \right)\ket{n}.
\end{align}
These states are also called Fock States and the Hilbert Space they span is called the Fock Space. The operators $\hat a^{\dagger}, \hat a$ are the 
creation and annihilation operators. We set $\hbar = 1$ from here on. The position and momentum quadratures are
\begin{align}
\hat x=\tfrac{1}{\sqrt2}(\hat a+\hat a^\dagger),\quad
\hat p=\tfrac{1}{i\sqrt2}(\hat a-\hat a^\dagger),\quad
[\hat x_j,\hat p_k]=i\delta_{jk}.
\end{align}
For $n$ modes, stack $\hat r=(\hat x_1,\hat p_1,\dots,\hat x_n,\hat p_n)^{\mathsf T}$ and
\begin{align}
\Omega=\bigoplus_{k=1}^n\begin{pmatrix}0&1\\-1&0\end{pmatrix}.
\end{align}

\subsubsection*{Gaussian States and Operations}
A state in a system of $n$ photonic field modes is Gaussian iff its Wigner function is multivariate Gaussian. It is fully specified by its first two moments, a mean vector $\bar{r} \in \mathbb{R}^{2n}$ and a covariance matrix $\sigma \in \mathbb{R}^{2n\times 2n}$:
\begin{align}
\bar r=\langle \hat r\rangle,\qquad
\sigma_{ij}=\tfrac12\langle\{\Delta\hat r_i,\Delta\hat r_j\}\rangle,
\end{align}
with uncertainty constraint $\sigma+i\Omega/2\ge0$.

Gaussian operations are those that preserve the Gaussianity of states. A general Gaussian channel acts as
\begin{align}
\sigma \to F \sigma F^{T} + N, \qquad \bar{r} \to F \bar{r} + \bar{d}.
\end{align}
A Gaussian unitary satisfies $N=0$ and $F\Omega F^{\mathsf T}=\Omega$ and we will omit $N$ if it is zero.

\subsubsection*{Coherent States}
Coherent States are used to describe classical laser light in Fock Space. They are the eigenstates of the annihilation operator, $\hat a\ket{\alpha} =\alpha \ket{\alpha}$, and their representation in the Fock Basis is
\begin{align}
\ket{\alpha} = e^{-\frac{\left| \alpha \right| ^{2}}{2}	}\sum_{n=0}^{\infty} \frac{\alpha^{n}}{\sqrt{ n! }} \ket{n}. 
\end{align}
The mean vector and covariance matrix of a single-mode coherent state are $\bar{r} = \sqrt{ 2 }(\mathrm{Re}(\alpha),\mathrm{Im}(\alpha))^T$, $\sigma = \frac{\mathbb 1}{2}$.
\subsubsection*{Thermal States}
Vacuum states that are ``broadened'' by thermal noise are called thermal vacuum states,
\begin{align}
    S_{E} = \frac{1}{E+1}\sum_{n=0}^{\infty} \left( \frac{E}{E+1} \right)^n |n\rangle\langle n|,
\end{align}
where $E$ is the average number of noise photons.

\subsubsection*{Two-Mode Squeezed Vacuum State}
A \gls{tmsv} state is a bipartite entangled state, defined by a complex parameter $\zeta = re^{i\theta}$, where $r$ is the squeeze magnitude and $\theta$ is the phase. The two-mode squeezing operation is defined by
\begin{align}
F= \begin{pmatrix}
\cosh(r)\mathbb{1}_{2} & \sinh(r) S_{\theta} \\
\sinh(r) S_{\theta} & \cosh(r) \mathbb{1}_{2}
\end{pmatrix}, \qquad \bar{d} = \mathbf{0},
\end{align}
where
\begin{align}
S_{\theta} = \begin{pmatrix}
\cos(\theta) & \sin(\theta) \\
\sin(\theta) & -\cos(\theta)
\end{pmatrix}.
\end{align}
We get a \gls{tmsv} state when we apply this operation to a two-mode vacuum state.

\subsubsection*{Beam Splitter}
The standard beam splitter with transmissivity $\eta$ mixes two bosonic field modes. It is defined as
\begin{align}
F = \begin{pmatrix}
\sqrt{ \eta }\, \mathbb{1}_{2} & \sqrt{1- \eta }\, \mathbb{1}_{2} \\
-\sqrt{ 1-\eta }\, \mathbb{1}_{2} & \sqrt{ \eta }\, \mathbb{1}_{2}
\end{pmatrix}, \qquad \bar{d} = 0.
\end{align}

\subsubsection*{Tracing Out}
Given a joint Gaussian state $\rho_{AB}$ in two systems $A$ and $B$, with
\begin{align}
\sigma_{AB} = \begin{pmatrix}
\sigma_{A}  & C \\
C^{T} & \sigma_{B}
\end{pmatrix}, \qquad \bar{r}_{AB} = \begin{pmatrix}
\bar{r}_{A} \\
\bar{r}_{B}
\end{pmatrix},
\end{align}
its reduced states $\rho_{A}$ and $\rho_{B}$ are Gaussian, with covariance matrices $\sigma_{A}$ and $\sigma_{B}$ and displacement vectors $\bar{r}_{A}$ and $\bar{r}_{B}$.

\section{Correlation Distribution Procedure}\label{sec:procedure}
We now show that the output of the \gls{tmsv} based correlation procedure is always correlated if the jammer is restricted to Gaussian states.
To make the notation more compact, we define $\eta' := 1-\eta$, $s_r := \sinh(2r)$ and $c_r := \cosh(2r)$. We start with a \gls{tmsv} state with $\theta = 0$
{\setlength{\arraycolsep}{4pt} \renewcommand{\arraystretch}{0.8}
\begin{align}
\sigma_{\mathrm{TMSV}} = \frac{1}{2}  \begin{pmatrix}
c_r & 0 & s_r& 0 \\
0 & c_r & 0 & -s_r\\
s_r &0 & c_r & 0 \\
0& -s_r& 0 & c_r
\end{pmatrix},\bar{r}_{\mathrm{TMSV}} = \begin{pmatrix}
0 \\ 0 \\ 0 \\ 0
\end{pmatrix},
\end{align}}
and a general Gaussian state $\tau$
\begin{align}
\sigma_{\tau} = \begin{pmatrix}
A &  C \\
C &  B
\end{pmatrix}, \qquad \bar{r}_{\tau}=\begin{pmatrix}
a \\
b
\end{pmatrix}.
\end{align}
Then
{\setlength{\arraycolsep}{3pt} \renewcommand{\arraystretch}{0.8}
\begin{align}
\sigma_{OUT} = \frac{1}{2} \begin{pmatrix}
c_{r} & 0 & \sqrt{ \eta }s_{r} & 0 \\
0 & c_{r} &  0 & -\sqrt{ \eta }s_{r}\\
\sqrt{ \eta }s_{r} & 0 & 2(\eta')A + \eta c_{r} & 2(\eta')C \\
0 & -\sqrt{ \eta }s_{r} & 2(\eta')C & 2(\eta')B + \eta c_{r}
\end{pmatrix}
\end{align}}
\begin{align}
    \bar{r}_{OUT} = \left(0,0,\sqrt{ \eta' }\,a,\sqrt{ \eta' }\,b \right)^T
\end{align}
Therefore, after homodyne measurement, we get two random variables $X$ and $Y$ with the joint probability density $P_{\hat{x}\hat{x}}^{AB}$, which is a bivariate Gaussian with
\begin{align}
\Sigma_{\hat x \hat x}= \frac{1}{2} \begin{pmatrix}
c_r & \sqrt{ \eta }s_r \\
\sqrt{ \eta }s_r  & 2(\eta')A + \eta c_r
\end{pmatrix}, \mu_{\hat x \hat x} = \begin{pmatrix}
0 \\
\sqrt{ \eta' }\, a
\end{pmatrix}.
\end{align}
The linear correlation coefficient of $X$ and $Y$ is
\begin{align}
\rho = \frac{\sqrt{ \eta }s_r}{\sqrt{ c_r(2(\eta')A + \eta c_r) }}.
\end{align}
Its sign does not depend on any parameter of $\tau$, it is always positive as long as $r>0$. We now perform sign binarization by defining random variables 
\begin{align}\label{def:binary-correlations}
    U = \mathrm{sgn}(X),\qquad V = \mathrm{sgn}(Y).
\end{align}
The outcomes correspond to the four quadrants of $\mathbb{R}^{2}$,
\begin{align}
q_{00} &= \int_{x_{A}<0}  \int_{x_{B}<0}   P_{\hat{x}\hat{x}}^{AB}(x_{A}, x_{B}) \, dx_{A} dx_{B} \end{align}
and $q_{01}$, $q_{10}$ and $q_{11}$ are defined analogously.
We bring the distribution $P_{\hat{x}\hat{x}}^{AB}$ into standard-normal form by setting
\begin{align}
\mathbf{z} = \begin{pmatrix}
z_{1} \\
z_{2}
\end{pmatrix} = \begin{pmatrix}
x_{A} \\
x_{B} - \sqrt{ \eta' }\,a
\end{pmatrix}
\end{align}
\begin{align}
\sigma_{A}^{2} = \tfrac{c_r}{2}, \qquad \sigma_{B}^{2} = \eta'A+\tfrac{\eta c_r}{2}
\end{align}
Then 
$\begin{pmatrix}
Z_{1}\\Z_{2}\end{pmatrix}=\begin{pmatrix} \frac{z_{1}}{\sigma_{A}} \\ \frac{z_{2}}{\sigma_{B}}.
\end{pmatrix}$ 
is a standard bivariate normal vector with correlation coefficient $\rho$ and the region corresponding to $q_{11}$ is $Z_{1}>0, Z_{2}> -b, b=\frac{\sqrt{ \eta' }\,a}{\sigma_{B}}$. In other words,
\begin{align}\label{eqn:structure-of-q-1}
q_{00} = & \,  \Phi_{2}\left( 0, -b;\rho \right), \\
q_{01}= & \, \tfrac{1}{2} - \Phi_{2}\left( 0, -b;\rho \right), \\
q_{10}= & \, \Phi(-b) - \Phi_{2}\left( 0, -b;\rho \right),  \\
q_{11} = & \, \tfrac{1}{2} - \Phi(-b) + \Phi_{2}\left( 0, -b;\rho \right),\label{eqn:structure-of-q-4}
\end{align}
where $\Phi_{2}$ is the cumulative distribution function (CDF) of the standard bivariate normal distribution, $\varphi_2$ is its density, and $\Phi$ is the CDF of the standard 1D Gaussian.

We want to use the binarized random variables $U,V$ for the correlated code, so it remains to show that $\rho>0$ implies $\rho_{\mathrm{bin}}>0$, where $\rho_{\mathrm{bin}}$ is the linear correlation coefficient of the binary variables.

For the case $a=0$ we can give an explicit expression in the form of the arcsine law \cite{Vleck1966}, giving us
\begin{align}
\rho_{\mathrm{bin}} = \tfrac{2}{\pi} \arcsin \rho,
\end{align}
which is positive whenever $\rho$ is positive. For $a \neq 0$ we do not give an exact formula, however in the Appendix we show that $\rho >0 \Longrightarrow \rho_{\mathrm{bin}}>0$. This directly implies that the mutual information $I(U, V)$ is positive for all Gaussian jammer states, including the specific case of the coherent and thermal states that the jammer can use in our jamming model \eqref{def:bpsk-avc}.

For all cases where the jammer has only a finite number of choices, this implies a uniform lower bound $I(U,V)>\delta'$ for a fixed $\delta'>0$ and all jammer choices.

\section{Main Result and Proof}
    Based on the correlation distribution procedure described in Section \ref{sec:procedure}, we obtain the following result.
    \begin{theorem}\label{thm:main}
        For all values $\alpha\in\mathbb R^+$, if the symbol $x=2$ creates a thermal state $S_{\alpha^2}$, then $C=0$. If $x=2$ creates a \gls{tmsv} state with marginals $S_{\alpha^2}$ and the sender is allowed to use homodyne detection on his half, then $C=\bar C$. In addition, $\bar C>0$ for all $\alpha \in \mathbb R^+$.
    \end{theorem}
    \begin{remark}
        The central observation in Theorem \ref{thm:main} is that a specific \emph{physical} channel, jammer and decoder setup yields positive capacity if the transmitter can create, transmit and access a specific entangled state, while the exact same setup has zero capacity when the transmitter can create and transmit the marginal of the \gls{tmsv} state.
    \end{remark}
    \begin{remark}
        Our result provides an improvement upon the work \cite{ahlswede-cai-correlated} by showing that, in the particular case at hand, already the \emph{arbitrarily varying} source created by the correlation procedure suffices to achieve $\bar C$.
    \end{remark}
    \begin{remark}
        Our result may be expressed as one of super-activation: The channel $\{2\}\to S_{\alpha^2}$ followed by homodyne detection clearly has $C=0$, even when it gets modified such that instead of $S_{\alpha^2}$ the sender creates a \gls{tmsv} state, sends one half to the receiver while keeping (and measuring) the other half (see e.g. \cite{nonlocality}). Just as well, the channel from \eqref{def:bpsk-avc} with $x\in\{0,1\}$ has $C=0$. If used together however, then $C>0$. This is reminiscent of the super-activation phenomenon for quantum channels discovered in \cite{Smith08}, where two channels with zero quantum capacity combine to yield positive capacity.
    \end{remark}
    To prove Theorem \ref{thm:main}, we first note down some preliminary statements. The first statement ensures that the correlations created by the correlation distribution procedure of Section \ref{sec:procedure} let any action of the jammer effectively be modeled as a \gls{bsc} between sender and receiver, where the jammer chooses the parameter of the \gls{bsc}. As a first step on this way, we show a geometric property of the set of correlations.
\subsubsection*{Geometric Properties}
    Given a fixed but arbitrary $\alpha$, we define $\qa$ to be the set of all distributions of $(U,V)$ as defined in \eqref{def:binary-correlations} and $\qa(\alpha)$ the subset of $\qa$ for which the jammer has jamming energy per symbol $\alpha^2$. We define further the distributions $q_c$, $q_0$ and $q_1$ as
    \begin{align}
        q_c(u,v)&:=\tfrac{1}{2}\delta(u,v)\\
        q_0&:=\pi\otimes\delta_0\\
        q_1&:=\pi\otimes\delta_1,
    \end{align}
    the set $\Delta:=\mathrm{conv}(\{q_c,q_0,q_1\})$ and the set $\Delta_\delta:=\mathrm{conv}(\{q_c,\tilde q_0,\tilde q_1\})$ with $\tilde q_i:=(1-\delta)q_i+\delta q_c$. The energy limit $\alpha^2$ of the jammer bounds the displacement and the covariance entries of the jamming state, so that for every $\alpha\neq0$ there is a $\delta>0$ such that $\qa(\alpha)\subset\Delta_\delta$. 
    \begin{proof}
        We first show $\qa\subset\Delta$. Each member $q$ of $\qa$ is given by the formulas in \eqref{eqn:structure-of-q-1} to \eqref{eqn:structure-of-q-4}.
        We embed $\mathcal P(\mathbf X)$ and thereby $\Delta$ into $\mathbb R^3$ by mapping any distribution $r$ to $r(0,0)e_0+r(0,1)e_1+r(1,0)e_2$. A normal to $\Delta$ is then given by  $(1,1,0)$. 
        The equality $\langle q-q_c,n\rangle =0$ shows that each $q\in\qa$ is also in the affine hull of $\Delta$. With \eqref{eqn:structure-of-q-1}-\eqref{eqn:structure-of-q-4} and $\rho>0$ we have $\lambda_c, \lambda_0, \lambda_1 \geq 0$, and therefore $\qa\subset\Delta$.
    \end{proof}
\subsubsection*{Effective Channel}
    Consider $w_0$ as in \eqref{def:bpsk-avc} with $p=0$ and $s,x\neq2$: $w_0(y|s,x)=\delta(x,y)$ if $s=x$, $w_0(y|s,x)=\tfrac{1}{2}$ ($s\neq x$). Then $w_0$ is symmetrizable with the choice $u(s|x)=\delta(s,x\oplus1)$. Will the use of a correlated random variable under partial control of the jammer still enable the sender and receiver to communicate? For $q\in\mathcal P(\{0,1\}^2)$, set
    \begin{align}
        \tilde w(y|q,s,x):=\sum_{u,v}w_0(y\oplus v|s,x\oplus u)q(u,v)
    \end{align}
    as in \eqref{def:augmented-avc}. If no restrictions on $q$ are imposed, the jammer would set $q(u,v)=\tfrac{1}{4}$ and the channel $\tilde w$ would be equal to $b_{1/2}$. By introducing the restriction $q\in\Delta_\delta$ for some $\delta>0$, we can however limit the impact of the jammer and show the property
    \begin{align}\label{eqn:bsc-property}
        \forall q\in\Delta_\delta\ \exists t\in[0,\tfrac{1}{2}-\tfrac{\delta}{4}]:\ \ \tilde w = b_t,
    \end{align}
    thus demonstrating that data transmission over $\tilde w$ is possible. 
    \begin{proof}
        We first show \eqref{eqn:bsc-property} for $\delta=0$. If $q=q_c$, then
        \begin{align}
            \tilde w_0(y|q_c,s,x)
                &= \tfrac{1}{2}\left(\tfrac{1}{2}+\delta(x,y)\right),
        \end{align}
        showing that in this case $\tilde w=b_{1/4}$. If $q=q_0$ or $q=q_1$, the randomization over $u$ and $v$ becomes \emph{independent} between sender and receiver, thus $\tilde w=b_{1/2}$. The map $q\to w(\cdot|q,\cdot,\cdot)$ is linear, thus every $\tilde w$ generated from $q\in\Delta$ is in $\mathrm{conv}(\{b_{1/4}, b_{1/2}\})$. Since the map $t\to b_t$ is linear itself, we get
        \begin{align}\label{eqn:w0-bsc}
            \tilde w_0(y|\lambda_cq_c+\textstyle\sum_{i=0}^1\lambda_iq_i,s,x) = b_{\lambda_c/4+(\lambda_0+\lambda_1)/2}(y|x).
        \end{align}
        Most importantly, $\tilde w(\cdot|q,s,\cdot)$ becomes independent from $s$ as long as $q\in\Delta$!
    \end{proof}
    From the above study it is immediate that $q\in\Delta_\delta$ implies $\tilde w=b_\xi$ for $\xi=\tfrac{1}{4}\lambda_c\delta +(1-\lambda_c)\left( \tfrac{1-\delta}{2} + \tfrac{\delta}{4} \right)$ where $\lambda_c\in[0,1]$, and that therefore 
    \begin{align}\label{eqn:lower-bound-on-t}
        t\leq\tfrac{1}{2}- \tfrac{\delta}{4} =:t_{s\neq2}.
    \end{align}
    For the case $s=2$, $\tilde w=b_t$ for some $t$ follows directly from \eqref{def:bpsk-avc}. Similar arguments as above show that in this case $t\leq\delta\tilde p + (1-\delta)/2=:t_{s=2}$. 
    Since $\delta>0$ and $\tilde p<1/2$ imply $\max\{t_{s=2},t_{s\neq2}\}<1/2$ we get $\tilde C(\tilde w)>0$ as soon as the restriction $q\in\Delta_\delta$ applies by invoking the code construction of \cite[Theorem 1]{ahlswedeWolfowitz}. For the reader's convenience, figure \ref{fig:delta-bow} visualizes $\qa$ and $\mathrm{conv}(\{q_c,q_0,q_1\})$ as part of the full probability space.
    \begin{figure}
        \includegraphics[trim={0cm 2.5cm 0 0cm},clip,width=.45\textwidth]{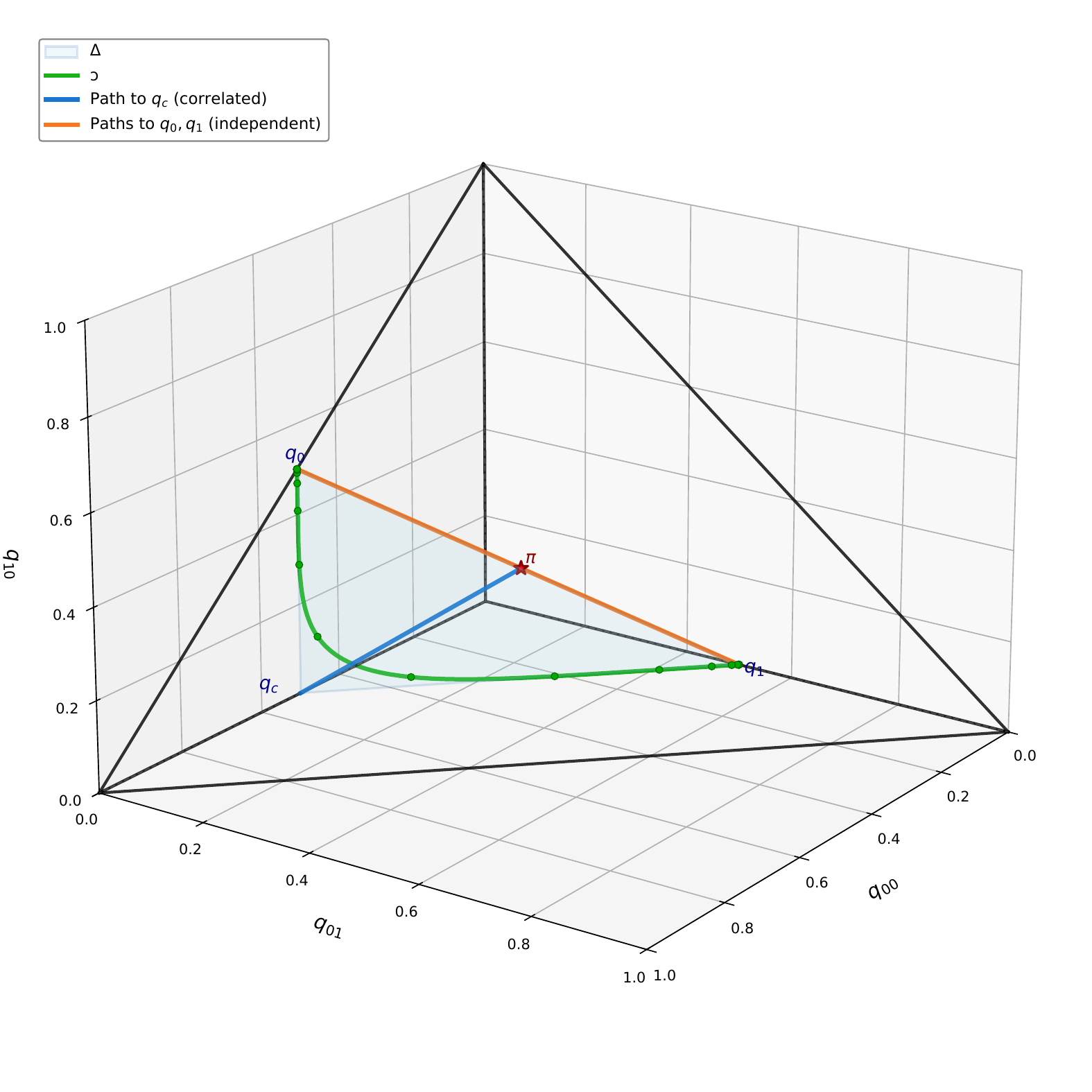}
        \caption{Depicted are $\qa$ and $\Delta$ as part of the convex polytope $\mathcal P(\{0,1\}^2)$, using $r=A=2$. We observe that $\qa$ approaches $q_0$ and $q_1$ in its limits.}\vspace{-3mm}
        \label{fig:delta-bow}
    \end{figure}
\begin{proof}[Proof of Theorem \ref{thm:main}]
    In order to fully capture the model \eqref{def:bpsk-avc} and finally prove the main theorem we first relax our earlier assumption $p=0$:
    Consider now $w$ instead of $w_0$, but keep $s,x\neq2$. We show
    \begin{align}\label{eqn:inside-bsc-family}
        w = b_{p}\circ w_0.
    \end{align}
    Let $s=x$. By \eqref{def:bpsk-avc} it holds $ w_p(y|x,x)=b_{p}(y|x)$. Let $s\neq x$. Then it holds 
    $w(y|x,x\oplus1)=b_{1/2}(y|x)$.\\
    Since $b_{t_1}\circ b_{t_2}=b_{t_1+t_2-2t_1t_2}$ and therefore $b_p\circ b_{1/2}=b_{1/2}$, statement \eqref{eqn:inside-bsc-family}
    is proven.\\
    By our previous discussion (see eq. \eqref{eqn:w0-bsc}), $\tilde w_0$ has the property $\tilde w_0=b_t$. Thus for every $q\in\Delta$
    \begin{align}
        \tilde w(y|s,x) 
            &= b_{p+t - 2pt}(y|x).
    \end{align}
    Since by \eqref{eqn:lower-bound-on-t} we have $t\leq(1-\delta)/2$ for $q\in\Delta_\delta$ and for every \gls{bsc} we have $\tilde C\geq \max_r I(r;b_t)=I(\pi,b_t)=1-h(t)$ where $h$ is the binary entropy, the result follows from \cite[Theorem 1]{ahlswedeWolfowitz} and by noting that the choice $s=2$ generates a \gls{bsc} as well:
    Sender and receiver first generate correlation for $\sim\log n$ transmissions using only $x=2$. Then they generate common randomness e.g. at a rate $\tilde C/2>0$ for another $\sim\log n$ transmissions using only $x\in\{0,1\}$, followed by data transmission at rates arbitrarily close to $\bar C$. The last two steps follow procedure established in \cite{ahlswede-cai-correlated} and used in \cite{boche13}. The capacity $\bar C$ is bounded away from zero as follows: Upon choosing a uniform input distribution $\pi$ and Pinsker's inequality, we have for every channel $\omega$ from $\{0,1\}$ to $\{0,1\}$ with joint distribution $\omega_\pi(y|x)=\frac{1}{2}\omega(y|x)$ the inequality
    \begin{align}
    	I(\pi,\omega)&\ge\tfrac{1}{2}\parallel\omega_\pi-\pi\otimes\omega(\pi)\parallel_1^2\\
    	&=\tfrac{1}{2}|\omega(0|0) - \omega(0|1)|^2.
    \end{align}
    For our distribution $\omega(y|x) = \sum_s\lambda_s w(y|s,x)$, so that 
    \begin{align}
    	I(\pi,\omega)\geq\tfrac{1}{2}\min\left\{\left(p-\tfrac{1}{2} \right) ^2, \left( 2\tilde p-1 \right)^2 \right\} > 0
    \end{align}
    by definition of $p$ and $\tilde p$.
\end{proof}


\appendix
Let $X, Y$ be random variables that are distributed according to a non-degenerate bivariate Gaussian distribution
\begin{align}
\begin{pmatrix}
X \\ Y
\end{pmatrix} \sim \, \mathcal{N}\left( \begin{pmatrix}
\mu_{X} \\ \mu_{Y}
\end{pmatrix}, \begin{pmatrix}
\sigma_{X}^{2} & \rho \sigma_{X} \sigma_{Y} \\ \rho \sigma_{X} \sigma_{Y}  &  \sigma_{Y}^{2}
\end{pmatrix} \right)  
\end{align}
with correlation coefficient $\rho$. Now consider the sign binarized variables $U = \mathrm{sgn}(X)$ and $V = \mathrm{sgn}(Y)$. We call their correlation coefficient
\begin{align}
\rho_\mathrm{bin} = \frac{\mathrm{Cov}(U, V)}{\sqrt{ \mathrm{Var(U)}\mathrm{Var}(V) }}.
\end{align}
We show that $\rho>0$ implies $\rho_{bin}>0$. 

Start by switching to standardized variables
\begin{align}
\bar X = \frac{X-\mu_{X}}{\sigma_{X}}, \qquad \bar Y = \frac{Y-\mu_{Y}}{\sigma_{Y}},
\end{align}
so that
\begin{align}
\begin{pmatrix}
\bar X \\ \bar Y
\end{pmatrix} \sim \, \mathcal{N} \left( \begin{pmatrix}
0 \\0
\end{pmatrix}, \begin{pmatrix}
1 & \rho \\
\rho & 1
\end{pmatrix} \right).
\end{align}
We also define $t_{X} = -\frac{\mu_{X}}{\sigma_{X}}$ and $t_{Y} = -\frac{\mu_{Y}}{\sigma_{Y}}$, so that $U=\mathrm{sgn}(\bar X-t_{X})$ and $V=\mathrm{sgn}(\bar Y-t_{Y})$.
Let $\varphi_{2}$ denote the probability density of the standard bivariate normal distribution, and let $\Phi$ and $\Phi_{2}$ be the cumulative distribution functions of the univariate and bivariate normal distributions. The sign of $\rho_{\mathrm{bin}}$ is determined by $\mathrm{Cov}(U,V)$, as the denominator is always positive. It can be expressed as
\begin{align}
\tfrac{1}{4} \mathrm{Cov}(U,V) = p_{++}-p_{+}^{U}p_{+}^{V},
\end{align}
where $p_{++}=\mathbb{P}(U=+1, V=+1)$, $p_{+}^{U}=\mathbb{P}(U=+1)$, $p_{+}^{V}=\mathbb{P}(V=+1)$. Explicitly we have
\begin{align}
p_{++}-p_{+}^{U}p_{+}^{V} = &  1-\Phi(t_{X})-\Phi(t_{Y})+\Phi_{2}(t_{X},t_{Y};\rho)\\&-(1-\Phi(t_{X}))(1-\Phi(t_{Y})) \nonumber \\
= & \Phi_{2}(t_{X},t_{Y};\rho) - \Phi(t_{X})\Phi(t_{Y}).
\end{align}
By definition, this expression is equal to zero for $\rho=0$. The derivative of $\Phi_{2}$ with respect to $\rho$ (Plackett's formula \cite{plackett1954}) is
\begin{align}
\frac{\partial}{\partial \rho}\Phi_{2}(t_{X},t_{Y};\rho) = \varphi_{2}(t_{X},t_{Y};\rho),
\end{align}
which for $\left| \rho \right|<1$ is strictly positive on $\mathbb{R}^{2}$. Therefore $\Phi_{2}(t_{X},t_{Y};\rho)$ and in effect $\mathrm{Cov}(U,V)$ is strictly increasing in $\rho$, implying
\begin{align}
\rho>0 \quad \Longrightarrow \quad  \rho_{\mathrm{bin}}>0.
\end{align}

\bibliographystyle{IEEEtran}
\bibliography{arXiv-bib}

\end{document}